\documentclass[journal,transvt]{IEEEtran}
\usepackage{amsfonts}
\usepackage{dsfont}
\usepackage{setspace}
\usepackage{color}
\usepackage{amssymb}
\usepackage{cite}
\usepackage[cmex10]{amsmath}
\usepackage{algorithm}
\usepackage{algorithmic} 
\usepackage{array}
\usepackage{mathrsfs}
\usepackage{graphicx}
\usepackage{latexsym}
\usepackage{amscd}
\usepackage{amsfonts}
\usepackage{subfigure}
\usepackage{amsmath,amscd,amssymb,verbatim}
\usepackage{graphics}
\usepackage{amsthm}
\usepackage[utf8]{inputenc}
\usepackage{authblk}
\usepackage{amsmath,amsthm}
\usepackage{blindtext}
\IEEEoverridecommandlockouts

\newtheorem{Proposition}{\underline{Theorem}}
\begin{document}

\title{On the Convex Properties of Wireless Power Transfer with Nonlinear Energy Harvesting
}
\author[$\dag$]{Yulin Hu,   Xiaopeng Yuan, Tianyu Yang, Bruno Clerckx
and Anke Schmeink 
}
\maketitle

\begin{abstract}
The convex property of a nonlinear   wireless power transfer (WPT)      is characterized in this work.
Following a nonlinear energy harvesting model, we express the relationship between  the harvested  direct current (DC)  power and  the power of  the received radio-frequency signal via an implicit  function, based on which  the convex property is further proved.  In particular, for a predefined rectifier's input signal distribution, we show that the harvested DC power of the nonlinear model is convex in the reciprocal of the rectifier's input signal power. Finally, we provide an  example to show  the advantages of applying the convex property in WPT network designs.
\end{abstract}

\section{Introduction}
Wireless power transfer (WPT) via radio-frequency (RF) radiation  has attracted significant attention in recent years. In particular, RF radiation has indeed become a viable source for energy harvesting (EH) with clear applications
in wireless sensor networks  in Smart City and  Internet of Things scenarios~\cite{backgournd}.
In the EH process, the received RF signal   is required to be converted into a direct current (DC)  signal.
Generally, this    EH model  (i.e., RF-to-DC conversion) has been considered as either a linear or a nonlinear process in the literature.
Under the assumption of a linear power EH model, various   works have been made to propose  optimal designs for   WPT networks, e.g., optimal scheduling~\cite{linear2} and resource allocation~\cite{linear3}.
Unfortunately,  in practice the RF-to-DC conversion is generally nonlinear, which    makes all the results conducted following the linear EH model  inaccurate~\cite{Clerckx_a,Clerckx_b,Clerckx_c ,Clerckx_d}.
In particular, the harvested DC power depends on the properties of the input signal (power and shape), while the classical linear model ignores much of such dependency.
 
A general nonlinearity of the EH model has been proposed in~\cite{Bruno_2016} via implicit equations, where the nonlinearity of the rectifier is characterized by the fourth and
higher order terms, which makes this model  more accurate and therefore widely-accepted. 
Following this nonlinear model,
a set of studies~\cite{Nonlinear1,Nonlinear2,Nonlinear3}  have provided  suboptimal resource/power allocation designs for the WPT network.
On the other hand, in comparison to the linear model, this nonlinearity makes the WPT network design problems  more complex and challenging, {as  the objective function or the constraint functions in these problems involve  implicit functions (representing the nonlinear   the RF-to-DC conversion), due to which the convex features of the problems are unlikely to be shown.}
Nevertheless, recently in~\cite{Bruno_2016} it is shown that from a signal design perspective,
maximizing the output DC current is equivalent to the  modified problem of treating the rectifier parameters (see $k_i$ in Equation (5) of~\cite{Bruno_2016})  constant.
 This highlights the dependency of the harvested DC power on the input signal and is
then leveraged in~\cite{Clerckx_TWC2018} to show the convexity of the diode current with respect to the input signal.
This convexity  is the reason why input distributions such as non-zero mean~\cite{Clerckx_e}, real Gaussian~\cite{Clerckx_f} and on-off keying~\cite{Clerckx_g} are favoured for simultaneous wireless information and power transfer (SWIPT).

In this paper, we build upon these observations and further analyze the convexity properties of the energy harvester in WPT, while taking the consideration of  the variableness of  the rectifier parameters. 
 In particular,  we   prove that under a predesigned     waveform (i.e., the distribution of the input signal is given),
 the harvested DC power via the nonlinear WPT model proposed in~\cite{Bruno_2016} is convex in the reciprocal of      the input signal power.
 This   is then shown using a simple example  to facilitate      the design (e.g., positioning, scheduling and so on) of  WPT networks.

The remaining of the paper is organized as follows. In Section~II, we first review the nonlinear WPT  model introduced in ~\cite{Bruno_2016}, following which we express an implicit equation representing the nonlinear relationship between  the harvested DC power and the power of the RF signal. Subsequently, we further characterized the convexity property of the nonlinear model~in Section III. In Section IV, we provide an example applying the introduced convex property in WPT network designs.
Finally, Section~IV provides our conclusions.

\section{Nonlinear Charging Model}
 To expresses the non-linearity of the diode,   by applying a Taylor expansion of  the   diode current~$i_{\rm d}$          the authors in \cite{Bruno_2016} show that  the following relationship between~$i_{\rm d}$ and the received/input signal $y_{\rm in}$ holds (see Eq.(6) in~\cite{Bruno_2016} for details) 
\begin{equation}
i_{\rm d}  = \sum\nolimits_{i=0}^\infty  {\cal K}_i (I_{\rm out})    R_{\rm ant}^{\frac{i}{2}}  {y_{\rm in}}^i,
\end{equation}
where  
 ${\cal K}_i (I_{\rm out}) $, $i=0,...,\infty$ are the rectifier characteristic functions with respect to the  rectifiers' output current $I_{\rm out}$. Specifically, ${\cal K}_i (I_{\rm out}) $ is defined as follows: For  i=0,  ${\cal K}_0 (I_{\rm out})  = {I_{\rm{s}}}\big( {{e^{ - \frac{{{I_{{\rm{out}}}}R}}{{n{v_t}}}}} - 1} \big)$, and for $i=1,...,\infty$,  it is given by ${\cal K}_i (I_{\rm out})  = {I_{\rm{s}}}\frac{{{e^{ - \frac{{{I_{{\rm{out}}}}R}}{{n{v_t}}}}}}}{{i!{{(n{v_t})}^i}}}$,
where $n$ is the ideality factor, and $v_t$ is the thermal voltage.

In particular, this nonlinear model in~\cite{Bruno_2016}  truncates the Taylor expansion to the $n_o$th order but
retains the fundamental non-linear behavior of the diode. After the truncation,  the    rectifiers' output current $I_{\rm out}$ is given by
\begin{equation}
I_{\rm out} \approx \sum\nolimits_{i=0}^{n_o} {\cal K}_i (I_{\rm out})   R_{\rm ant}^{\frac{i}{2}}\mathbb{E}\{{y_{\rm in}}^i\},
\end{equation}
where $R_{\rm ant}$ is the antenna impedance. 
According to Equation (19) in~\cite{Bruno_2016}, the following relationship holds
\begin{equation}
\label{eq:bruno_truncated}
\begin{split}
 e^{\frac{R_{\rm L} I_{\rm out}}{n v_t}}(I_{\rm out}+I_{\rm s})
&  \approx
  I_{\rm s}+\sum_{i~\text{
  even},i\geq 2}^{n_o}{\bar k}_i  R_{\rm ant}^{\frac{i}{2}}\mathbb{E}\{{y_{\rm in}}^i\},
\end{split}
\end{equation}
where    $I_{\rm s}$ is the reverse bias saturation current and    $n_o$ (even) is the truncation order.
In addition, ${\bar k}_i =\frac{I_{\rm s}}{i! (nv_t)^i}$ for $i\geq 2$  and with $i$ even, are the rectifiter characteristic constants, i.e., not influenced by $I_{\rm out}$. 

Denote by $R_{\rm L}$ the load resistance, then  the harvested DC power $P_{\rm dc}$ can be obtained by 
\begin{equation}
\label{eq:P_I}
P_{\rm dc}=I_{\rm out}^2 R_{\rm L}.
\end{equation}

We consider a network under a predesigned     waveform, i.e., the distribution of the input signal $y_{\rm{in}}$ is  given and hence  the $i$-th moment of  $y_{\rm{in}}$ can be known\footnote{
On Table III of \cite{Clerckx_TWC2019},   input signals with different distributions/modulations are discussed.}.
Then,
  the harvested DC power $P_{\rm dc}$ and the power $Q_{\rm rf}$ of the received RF signal can be presented by a nonlinear implicit function $\mathcal{F}_{\rm nl}$, i.e.,  $P_{\rm dc}=\mathcal{F}_{\rm nl}(Q_{\rm rf})$.
  In particular, the power of the received (input)  signal can be obtained by     $Q_{\rm rf} = \mathbb{E}\{y_{\rm in}^2\}$.
In addition, with the distribution of  $y_{\rm{in}}$,  we can draw that $\mathbb{E}\{{y_{\rm in}}^i\}$ ($i$ is even) is proportional to $Q_{\rm rf}^{i/2}$, given by
\begin{equation}
\label{eq:even_i_y_in}
\mathbb{E}\{{y_{\rm in}}^i\}=\lambda_i Q_{\rm rf}^{i/2},~i~\text{even},
\end{equation}
where 
 $\lambda_i $ is the waveform factor  with a unit power, given by ${\lambda _i} = \frac{{{\mathbb E}\{ {y_{{\rm{in}}}}^i\} }}{{{{\left( {{\mathbb E}\{ {y_{{\rm{in}}}}^2\} } \right)}^{\frac{i}{2}}}}}$. 

Combining~\eqref{eq:even_i_y_in} to~\eqref{eq:bruno_truncated}, we conduct the following relationship between $I_{\rm out}$ and  $Q_{\rm rf}$
\begin{equation}
\label{eq:sum_loops}
e^{\frac{R_{\rm L} I_{\rm out}}{n v_t}}(I_{\rm out}+I_{\rm s}) \approx \sum\nolimits_{j=0}^{n_o'} \alpha_j Q_{\rm rf}^j,
\end{equation}
where $n_o'=n_o/2$,
$\alpha_0=I_{\rm s}>0$ and $\alpha_j={\bar k}_{2j}  R_{\rm ant}^{j} \lambda_{2j}>0,~j\geq 1$. Clearly, the charged current $I_{\rm out}$ is an implicit function of $Q_{\rm rf}$.
Based on~\eqref{eq:P_I}, we further conclude that the harvested DC power $P_{\rm dc}$ is also an implicit function of $Q_{\rm rf}$, which can be expressed as  $P_{\rm dc}=\mathcal{F}_{\rm nl}(Q_{\rm rf})$.
So far, we have reviewed the nonlinear WPT  model introduced in~\cite{Bruno_2016} and discussed the relationship between the harvested DC power $P_{\rm dc}$ and the  $Q_{\rm rf}$  power of the   received    signal.
In the next section, we further characterize the convex property of the above WPT model.
\section{Convex Property of the Nonlinear WPT}
The harvested DC power in the EH process $P_{\rm dc}=\mathcal{F}_{\rm nl}(Q_{\rm rf})$ is an implicit function of  the power of the   received  RF signal $Q_{\rm rf}$, while  $Q_{\rm rf}$ can be seen as the interface between the EH process and the RF signal transmission process.
However, the relationship between $P_{\rm dc}$ and  $Q_{\rm rf}$ so far has been only  characterized implicitly by a nonlinear implicit  function $\mathcal{F}_{\rm nl}$, which introduces significant difficulties to  maximize~the~harvested DC power by applying optimal designs in the  RF signal transmission process. To address this issue, we show the convex property of the nonlinear WPT model in the following.

First,  we introduce a variable $u$ such that   $Q_{\rm rf}$  is modeled as a function of~$u$, where $u$ could be a  variable/factror  considered in the 
RF signal transmission process\footnote{For instance,  under a predesigned waveform ($y_{\rm{in}}$ has a given distribution),
  if we let $u$ denote the square of the distance between transmitter and receiver,  the received RF power $Q_{\rm rf}$ (in a free space channel) can be modeled as $Q_{\rm rf}=\frac{Q_0}{u}$, where $Q_0$ is the received power at unit distance. Except that, several factors, e.g., the transmit power at the transmitter and the gain of the channel selected for WPT, significantly influence $Q_{\rm rf}$. Depending on the   system design problem, one can model $u$ as one of the above factors or a function/combination of some of the factors.}.    
In addition, to facilitate our proof we define by $\rho(u)$ the right side of~\eqref{eq:sum_loops}, i.e., 
\begin{equation}\label{eq:rho_definition}
\rho(u)\buildrel \Delta \over =\sum\nolimits_{j=0}^{n_o'} \alpha_j \left( Q_{\rm rf}(u)\right)^j.
\end{equation}
Recall that $\alpha_j>0$ for $j=0,..,n_o'$ and the received  RF signal power has a  positive value $ Q_{\rm rf}(u)> 0$.
Hence, we have $\rho(u) > 0$.
According to~\eqref{eq:sum_loops}, it also holds that
\begin{equation}
e^{\frac{R_{\rm L} I_{\rm out}}{n v_t}}(I_{\rm out}+I_{\rm s})=\rho(u),
\end{equation}
where $I_{\rm out}$ is an implicit function of  ${\rho (u)}$. Hence, $I_{\rm out}$ is also a function of $u$.  We define this function by
$I_{\rm out}={{\mathop{\mathcal I}\nolimits} _{\rm out}}\left( {\rho (u)} \right)$.

We have the following key proposition addressing the     convexity of function ${{\mathop{\mathcal I}\nolimits} _{\rm out}}$.

\begin{Proposition}
 \label{le:Proposition1}
 ${\cal I}_{\rm out}\left(\rho(u)\right)$ and $P_{\rm dc}$ are convex in $u$, if the following inequality holds
\begin{equation}\label{eq:rho_condition}
{{\mathop \rho \limits^{..} (u)}-{\mathop \rho \limits^. (u)}^2\frac{1}{\rho(u)}}\geq 0,
\end{equation}
where ${\mathop \rho \limits^{.} (u)}$ and ${\mathop \rho \limits^{..} (u)}$ are the first order and second order derivatives of $\rho  (u)$ to $u$.
\end{Proposition}
\begin{proof}
The first order derivative of ${{\mathop{\mathcal I}\nolimits} _{\rm out}}\left( {\rho (u)} \right)$ to $u$ is given~by
\begin{eqnarray}
\label{eq:1st_derivative}
{{\mathop {\mathop{\mathcal I}\nolimits} \limits^.} _{\rm out}}\left( {\rho (u)} \right) = \frac{{\mathop \rho \limits^. (u)}}{{\frac{{{R_{\rm L}}}}{{n{v_t}}}\rho (u) + {e^{\frac{{{R_{\rm L}}{{\mathop{\mathcal I}\nolimits} _{\rm out}}\left( {\rho (u)} \right)}}{{n{v_t}}}}}}}.
\end{eqnarray}
Based on~\eqref{eq:1st_derivative}, the second order derivative can be obtained, which is provided in~\eqref{eq:2nd_derivative} on the top of next page.

\begin{table*}[!h]
\begin{equation}
\label{eq:2nd_derivative}
\begin{split}
{{\mathop {\cal I}\limits^{..}} _{{\rm{out}}}}\left( {\rho (u)} \right) &= \frac{\mathop \rho \limits^. (u) - \frac{{{R_{\rm{L}}}}}{{n{v_t}}}{{\mathop {\cal I}\limits^. }_{{\rm{out}}}}{{\left( {\rho (u)} \right)}^2}\left(\frac{{{R_{\rm{L}}}}}{{n{v_t}}}\rho (u) + 2{e^{\frac{{{R_{\rm{L}}}{{\cal I}_{{\rm{out}}}}\left( {\rho (u)} \right)}}{{n{v_t}}}}}\right)}{\frac{{{R_{\rm{L}}}}}{{n{v_t}}}\rho (u) + {e^{\frac{{{R_{\rm{L}}}{{\cal I}_{{\rm{out}}}}\left( {\rho (u)} \right)}}{{n{v_t}}}}}}\\
&= \frac{1}{{\frac{{{R_{\rm{L}}}}}{{n{v_t}}}\rho (u) + {e^{\frac{{{R_{\rm{L}}}{{\cal I}_{{\rm{out}}}}(\rho (u))}}{{n{v_t}}}}}}}\left( {\mathop \rho \limits^{..} (u) - \frac{{{R_{\rm{L}}}\mathop \rho \limits^. {{(u)}^2}}}{{n{v_t}}} \cdot \frac{{\frac{{{R_{\rm{L}}}}}{{n{v_t}}}\rho (u) + 2{e^{\frac{{{R_{\rm{L}}}{{\cal I}_{{\rm{out}}}}\left( {\rho (u)} \right)}}{{n{v_t}}}}}}}{{{{\left( {\frac{{{R_{\rm{L}}}}}{{n{v_t}}}\rho (u) + {e^{\frac{{{R_{\rm{L}}}{{\cal I}_{{\rm{out}}}}\left( {\rho (u)} \right)}}{{n{v_t}}}}}} \right)}^2}}}} \right).
\end{split}
\end{equation}
\hrule
\end{table*}


Note that it holds
\begin{eqnarray}
\label{eq:ineq_derivative1}
\begin{split}
\!\!\!\!\!\!&~{\left( {\frac{{{R_{\rm{L}}}}}{{n{v_t}}}\rho (u) + {e^{\frac{{{R_{\rm{L}}}{{\cal I}_{{\rm{out}}}}(\rho (u))}}{{n{v_t}}}}}} \right)^{\!\!2}}\\
 = &~ {\left( {\frac{{{R_{\rm{L}}}}}{{n{v_t}}}\rho (u)} \!\right)^{\!\!\!2}} \!\!+ 2\frac{{{R_{\rm{L}}}}}{{n{v_t}}}\rho (u){e^{\frac{{{R_{\rm{L}}}{{\cal I}_{{\rm{out}}}}(\rho (u))}}{{n{v_t}}}}} \!\!+ {e^{\frac{{2{R_{\rm{L}}}{{\cal I}_{{\rm{out}}}}(\rho (u))}}{{n{v_t}}}}}\\
 > &~ {\left( {\frac{{{R_{\rm{L}}}}}{{n{v_t}}}\rho (u)} \right)^{\!\!\!2}}\!\! + 2\frac{{{R_{\rm{L}}}}}{{n{v_t}}}\rho (u){e^{\frac{{{R_{\rm{L}}}{{\cal I}_{{\rm{out}}}}(\rho (u))}}{{n{v_t}}}}} \\
 = &~  \rho (u)\frac{{{R_{\rm{L}}}}}{{n{v_t}}}\left( {\frac{{{R_{\rm{L}}}}}{{n{v_t}}}\rho (u) + 2{e^{\frac{{{R_{\rm{L}}}{{\cal I}_{{\rm{out}}}}(\rho (u))}}{{n{v_t}}}}}} \right).
\end{split}
\end{eqnarray}

As ${\rho(u)}>0$, according to \eqref{eq:ineq_derivative1} we have
\begin{eqnarray}
\label{eq:ineq_derivative2}
\frac{\frac{R_{\rm L}}{n v_t}\left(\frac{R_{\rm L}}{n v_t}\rho (u)+2e^{\frac{R_{\rm L} {\cal I}_{\rm out}(\rho(u))}{n v_t}}\right)}{\left(\frac{R_{\rm L}}{n v_t}\rho (u)+e^{\frac{R_{\rm L} {\cal I}_{\rm out}(\rho(u))}{n v_t}}\right)^{\!\!\!2}} &<& \frac{1}{\rho(u)}.
\end{eqnarray}

Combining~\eqref{eq:2nd_derivative} with \eqref{eq:ineq_derivative2}, we have
\begin{eqnarray}
{{\mathop {\cal I}\limits^{..}} _{{\rm{out}}}}\left( {\rho (u)} \right)  &>& \frac{{\mathop \rho \limits^{..} (u)}-{\mathop \rho \limits^. (u)}^2\frac{1}{\rho(u)}}{\frac{R_{\rm L}}{n v_t}\rho (u)+e^{\frac{R_{\rm L}{\cal I}_{\rm out}\left(\rho(u)\right)}{n v_t}}} .
\end{eqnarray}

Hence, if ${{\mathop \rho \limits^{..} (u)}-{\mathop \rho \limits^. (u)}^2\frac{1}{\rho(u)}}\geq 0$ holds,
${\cal I}_{\rm out}\left(\rho(u)\right)$  is  convex with respect to variable $u$. 
Noting that the output current is definitely non-negative, i.e., $I_{\rm out}\geq 0$ holds. 
According to \eqref{eq:P_I},   $P_{\rm dc}$ is also  convex in $u$ if ${{\mathop \rho \limits^{..} (u)}-{\mathop \rho \limits^. (u)}^2\frac{1}{\rho(u)}}\geq 0$ holds.
\end{proof}
%
Furthermore, based on the convexity proved in Theorem \ref{le:Proposition1} we can derive out a more visualized sufficient condition.
\begin{Proposition}
 \label{le:Proposition3}
 The inequality \eqref{eq:rho_condition} holds if 
\begin{equation}\label{eq:sufficient_condition}
{{{\mathop Q \limits^{..}}_{\rm rf} (u)}Q_{\rm rf}(u)-\left({{\mathop Q \limits^.} _{\rm rf} (u)}\right)^2}\geq 0,
\end{equation}
holds, where ${{\mathop Q \limits^{.}}_{\rm rf} (u)}$ and ${{\mathop Q \limits^{..}}_{\rm rf} (u)}$ are the first order and second order derivatives of $Q_{\rm rf}  (u)$ to $u$.
\end{Proposition}
\begin{proof}
From the definition \eqref{eq:rho_definition} of $\rho(u)$, it can be easily derived that $\rho(u)>0$. Thus, the inequality \eqref{eq:rho_condition} is equivalent to ${\mathop \rho \limits^{..} (u)}\rho(u)\geq{\mathop \rho \limits^. (u)}^2$. Let ${\mathop Q \limits^{.}}_{\rm rf}={\mathop Q \limits^{.}}_{\rm rf} (u)$ and ${\mathop Q \limits^{..}}_{\rm rf}={\mathop Q \limits^{..}}_{\rm rf} (u)$. Then, we can derive that
\begin{equation*}
\begin{split}
\mathop\rho \limits^. (u)&=\frac{d\rho}{d Q_{\rm rf}}{{\mathop Q \limits^{.}}_{\rm rf}} \\
&=\sum\nolimits_{j=1}^{n_o'} \alpha_j j Q_{\rm rf}^{j-1}{{\mathop Q \limits^{.}}_{\rm rf}},\\
\mathop \rho \limits^{..} (u)&=\frac{d^2\rho}{d Q_{\rm rf}^2}{\mathop Q \limits^{.}}_{\rm rf}^2+\frac{d\rho}{d Q_{\rm rf}}{\mathop Q \limits^{..}}_{\rm rf} \\
&=\sum\nolimits_{j=2}^{n_o'} \alpha_j j(j\!-\!1)Q_{\rm rf}^{j-2}{\mathop Q \limits^{.}}_{\rm rf}^2+\sum\nolimits_{j=1}^{n_o'} \alpha_j jQ_{\rm rf}^{j-1}{\mathop Q \limits^{..}}_{\rm rf} \\
&= \alpha_1{\mathop Q \limits^{..}}_{\rm rf}+\sum\nolimits_{j=2}^{n_o'} \alpha_j jQ_{\rm rf}^{j-2}\left((j\!-\!1){\mathop Q \limits^{.}}_{\rm rf}^2+Q_{\rm rf}{\mathop Q \limits^{..}}_{\rm rf}\right)\\
&\geq  \alpha_1{\mathop Q \limits^{..}}_{\rm rf}+\sum\nolimits_{j=2}^{n_o'} \alpha_j j^2Q_{\rm rf}^{j-2}{\mathop Q \limits^{.}}_{\rm rf}^2 \\
&\geq 0~.
\end{split}
\end{equation*}
Based on the expressions of $\mathop\rho \limits^. (u)$ and $\mathop \rho \limits^{..} (u)$, we have
\begin{align}
& {\mathop \rho \limits^{..} (u)}(\rho(u)-\alpha_0) \nonumber\\
\geq & \left(\alpha_1{\mathop Q \limits^{..}}_{\rm rf}+\sum\nolimits_{j=2}^{n_o'} \alpha_j j^2Q_{\rm rf}^{j-2}{\mathop Q \limits^{.}}_{\rm rf}^2\right)\left(\sum\nolimits_{j=1}^{n_o'} \alpha_jQ_{\rm rf}^j \right) \\
\geq &  \left(\alpha_1\sqrt{Q_{\rm rf}{\mathop Q \limits^{..}}_{\rm rf}}+\sum\nolimits_{j=2}^{n_o'} \alpha_j j Q_{\rm rf}^{j-1}|{{\mathop Q \limits^{.}}_{\rm rf}}|\right)^2 \\
\geq&  \left(\alpha_1|{{\mathop Q \limits^{.}}_{\rm rf}}|+\sum\nolimits_{j=2}^{n_o'} \alpha_j j Q_{\rm rf}^{j-1}|{{\mathop Q \limits^{.}}_{\rm rf}}|\right)^2  \\
=&\left(\sum\nolimits_{j=1}^{n_o'} \alpha_j j Q_{\rm rf}^{j-1}{{\mathop Q \limits^{.}}_{\rm rf}}\right)^2 \nonumber \\
=& ~{\mathop \rho \limits^. (u)}^2, \nonumber
\end{align}
where the inequality between (16) and (17) holds according to the Cauchy-Buniakowsky-Schwarz Inequality, and the inequality between (17) and (18) is due to the condition in~\eqref{eq:sufficient_condition}.
Therefore, we can get
\begin{equation}
{\mathop \rho \limits^{..} (u)}-{\mathop \rho \limits^. (u)}^2\frac{1}{\rho(u)}\geq \frac{\alpha_0{\mathop \rho \limits^{..} (u)}}{\rho(u)}\geq 0.
\end{equation}
\end{proof}

Condition \eqref{eq:sufficient_condition} is a sufficient condition of \eqref{eq:rho_condition}. Thus, \eqref{eq:sufficient_condition} can also result in the convexity of ${\cal I}_{\rm out}\left(\rho(u)\right)$ and $P_{\rm dc}$ in~$u$.

Next, {we consider a special type of function  $Q_{\rm rf}(u)$, given by  $Q_{\rm rf}(u)=\frac{a}{u}$, where $a$ is a constant. With such function type, it can be easily proved that ${{{\mathop Q \limits^{..}}_{\rm rf} (u)}Q_{\rm rf}(u)-{{\mathop Q \limits^.} _{\rm rf} (u)}^{\!_2}}>0$ for $u>0$.}
According to Theorem~\ref{le:Proposition1} and Theorem~\ref{le:Proposition3}, $P_{\rm dc}$ is convex with respect to $u>0$.
Combining this example $(a=1)$ with Theorem~\ref{le:Proposition1} and Theorem~\ref{le:Proposition3}, we have
\begin{Proposition}
 \label{le:Proposition2} 
{Under a predesigned waveform (given distribution of the input signal $y_{\rm{in}}$),}
 ${\cal I}_{\rm out}\left(Q_{\rm rf}\right)$ and $P_{\rm dc}$ are convex in $\frac{1}{Q_{\rm rf}}$ for $Q_{\rm rf} >0 $.
\end{Proposition}
{
\begin{proof}
	Simply let $u=\frac{1}{Q_{\rm rf}}$, i.e. $Q_{\rm rf}=\frac{1}{u}$. Hence, $u>0$ holds.  It can be easily proved that $Q_{\rm rf}=\frac{1}{u}$ satisfies the condition in~\eqref{eq:sufficient_condition}. ${\cal I}_{\rm out}\left(Q_{\rm rf}\right)$ and $P_{\rm dc}$ are convex in $u$, i.e. convex in $\frac{1}{Q_{\rm rf}}$.
\end{proof}
}

We validate our analytical model in Fig.~\ref{fig:validation}, where different type of $y_{\rm{in}}$ are considered as per reference \cite{Clerckx_TWC2019}.
Clearly, the results match well with Theorem~\ref{le:Proposition2}.
\begin{figure}[!t]
    \centering
\includegraphics[width=0.47\textwidth,trim=0 17 0 0]{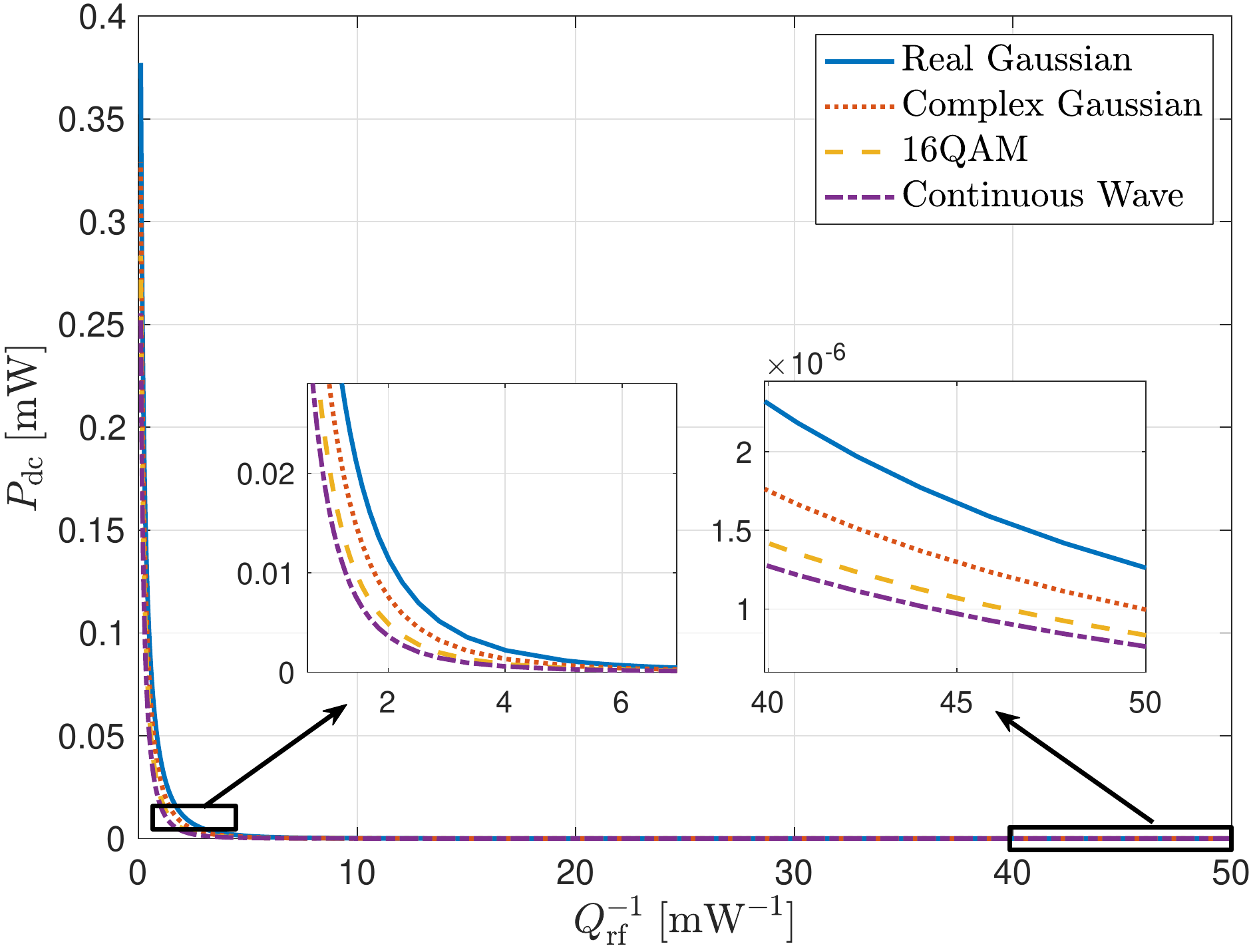}
\caption{Numerical results confirm Theorem~\ref{le:Proposition2}.}
\label{fig:validation}
\end{figure}
According to Theorem~\ref{le:Proposition2},  the harvested DC power is convex in $\frac{1}{Q_{\rm rf}}$.
Note that $Q_{\rm rf}$ is the received RF power and $\frac{1}{Q_{\rm rf}}$ is linear in the path-loss of the wireless link transmitting the RF signal, i.e., $Q_{\text{rf}} = z / d$,  $d$ is the path-loss and $z$ is a weight (e.g., due to channel fading) constant to $d$. 
 Hence, we  conclude that   the harvested DC power is convex in the path-loss of the RF transmission link. This property likely facilitates the optimal position design for the source, especially when the source provides energy supply for multiple users at different locations.
Moreover, the theorem  indicates the convexity between the harvested DC power and   the reciprocal of the transmit power of the RF signal, which provides guidelines for   power allocation designs.
\section{An Application:   WPT Transmitter Positioning   }
 
In this section, we provide a case study  to present the advantage of the proved convex property in solving optimization problems in practical WPT system designs.
Specifically, we consider a WPT system with a WPT transmitter at position  $(x,y)$ and a set of $N$  randomly located WPT receivers with positions  $\{(x_n, y_n)\}$, $n\in \mathcal{N}=\{1,...,N\}$. The received RF power of the $n$-th receiver is $Q_{\text{rf},n} = Q_0 / d_n$, where $Q_0$ is the received RF power at a unit distance and $d_n = (x-x_n)^2+(y-y_n)^2$ is the pathloss of the $n$-th receiver.  
{Taking both the   charging performance and  fairness into account,} we  consider to  maximize the minimal harvested DC power among all receivers by optimizing the transmitter's position $(x,y)$. Formally, the resulting  optimization problem is 
\begin{subequations}  \label{eq:opt}
\begin{align}
\underset{x, y, \tilde{P}_{\text{dc}}}{\text{max}} \;\;  &\;\;  \tilde{P}_{\text{dc}}    \label{opt_a}  \\
{\text{s.t.}} \;\;\;\;    &  \tilde{P}_{\text{dc}} - P_{\text{dc},n}\left(Q_{\text{rf},n}\right) \leq 0, \;\; \forall n \in \mathcal{N}, \label{opt_b}\\
&  x_{\text{min}} \leq x \leq x_{\text{max}}, \label{opt_c}\\
&  y_{\text{min}} \leq y \leq y_{\text{max}}, \label{opt_d}
\end{align}
\end{subequations}
where $P_{\text{dc},n}\left(Q_{\text{rf},n}\right)$ is the harvested DC power of the $n$-th receiver and $\tilde{P}_{\text{dc}}$ is the lower bound of the harvested DC power among all receivers. Moreover, $x_{\text{min}},y_{\text{min}}$ and $x_{\text{max}},y_{\text{max}}$ are the minimal and maximal values of the positions $(x_n, y_n), \forall n \in \mathcal{N}$, respectively.

Note that the problem in (\ref{eq:opt}) is non-convex due to the concave term $-P_{\text{dc},n}$ in (\ref{opt_b}). On the other hand, according to Theorem~\ref{le:Proposition2} $P_{\text{dc},n}$ is convex in $\frac{1}{Q_{\text{rf},n}}$, i.e., {also convex in}  the path-loss $d_n$.
Hence, the successive inner approximation (SIA) method \cite{SIA} can be applied to iteratively solve the convex-approximate problem. Specifically, in the $i$-th iteration, let $d_n^{[i]} = (x^{[i]}-x_n)^2+(y^{[i]}-y_n)^2$ and $\alpha_n^{[i-1]}=-P_{\text{dc},n}'(d_n^{[i-1]}) > 0$,  the first-order Taylor
approximation of term $-P_{\text{dc},n}$ over the point of $d_n^{[i-1]}$ is obtained as
\begin{align}
    -P_{\text{dc},n}(d_n^{[i]}) &\leq \alpha_n^{[i-1]}d_n^{[i]} -\alpha_n^{[i-1]}d_n^{[i-1]}- P_{\text{dc},n}(d_n^{[i-1]}) \\
    &\triangleq \hat{P}_{\text{dc},n}^{[i]}(\mathcal{V}^{[i]},\mathcal{V}^{[i-1]}),
\end{align}
where $\mathcal{V}^{[i]}= (x^{[i]},y^{[i]})$. Note that $\alpha_n^{[i-1]} > 0$ and $d^{[i]}_n$ is jointly convex over $(x^{[i]},y^{[i]})$, thus $\hat{P}_{\text{dc},n}^{[i]}$ is jointly convex over $(x^{[i]},y^{[i]})$. Then, the problem in (\ref{eq:opt}) is solved iteratively until the stable point is achieved. In the $i$-th iteration the approximated problem is written as

\begin{subequations}  \label{eq:opt_convex}
\begin{align}
\underset{\mathcal{V}^{[i]}, \tilde{P}_{\text{dc}}^{[i]}}{\text{max}} \;\;  &\;\;  \tilde{P}_{\text{dc}}^{[i]}   \\
{\text{s.t.}} \;\;\;\;    &  \tilde{P}_{\text{dc}}^{[i]} + \hat{P}_{\text{dc},n}^{[i]}(\mathcal{V}^{[i]},\mathcal{V}^{[i-1]}) \leq 0, \;\; \forall n \in \mathcal{N},\\
&  (\ref{opt_c}), (\ref{opt_d}). \nonumber
\end{align}
\end{subequations}

\begin{figure}[!t]
    \centering
\includegraphics[width=0.47\textwidth,trim=0 15 0 0]{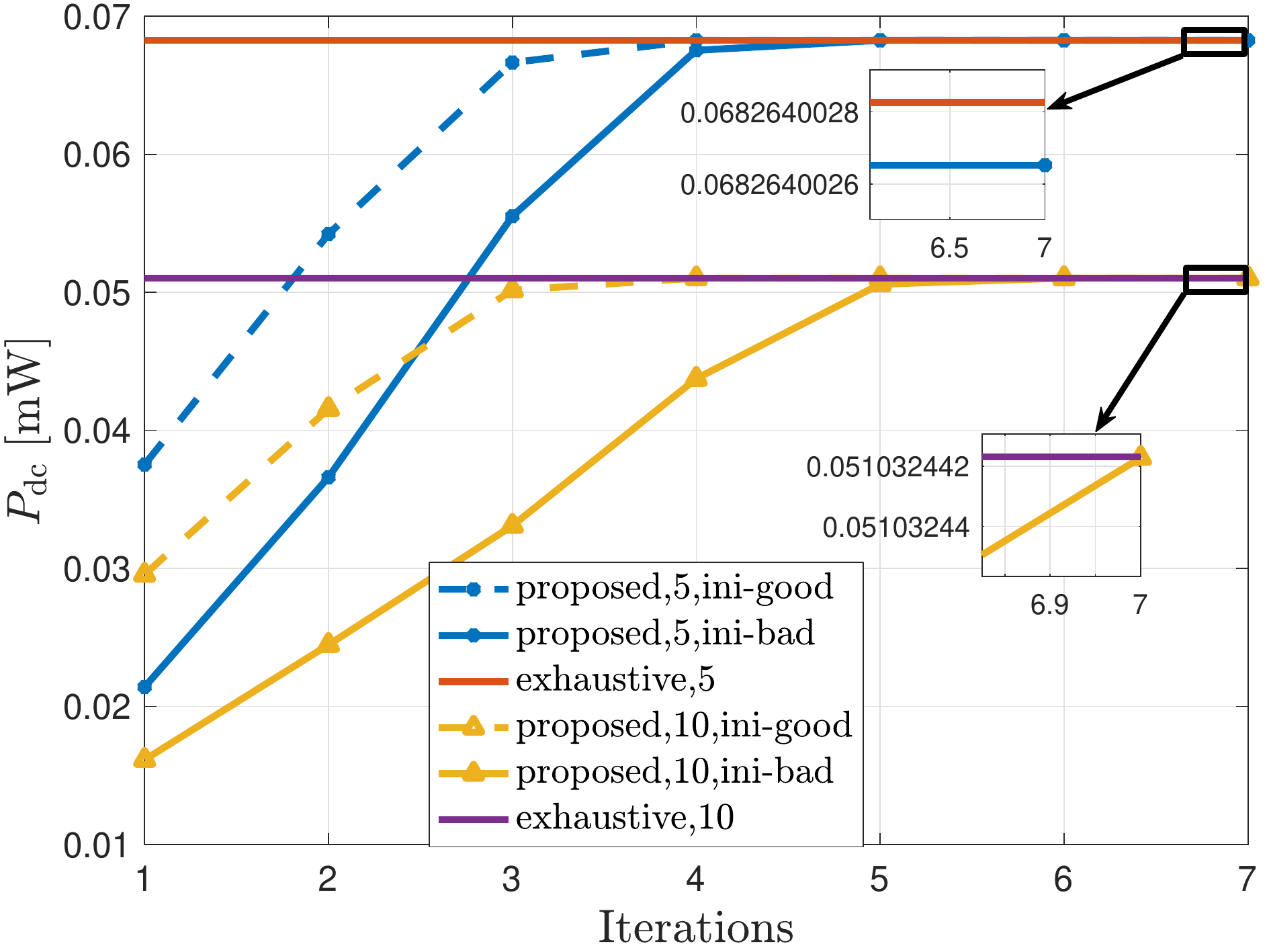}
\caption{Convergence behavior of proposed algorithm with 5 receivers and 10 receivers.}
\label{fig:converge}
\end{figure}

{
In the simulation, the WPT receivers are randomly located in a square area with the width of 5m. The transmit power of the WPT transmitter is set to 30dBm. The received RF power at a unit distance is set to 10dBm, i.e., $Q_0 = 10$dBm. In order to evaluate the performance of our iterative algorithm, we compare the results with the  optimum obtained by a grid based  exhaustive search. The exhaustive search is conducted as follows. First, we define the searching area as $\mathcal{A} := \{(x,y) | x_\text{min} \leq x \leq x_\text{max}, y_\text{min} \leq y \leq y_\text{max}\}$. Then, we discretize the area $\mathcal{A}$ into meshes with the resolution of $\xi$ and get a grid based searching area defined as $\tilde{\mathcal{A}} = \{(x_i,y_i) | i\in \mathcal{I}\}$, where $(x_i,y_i)$ is the $i$-th grid point and $\mathcal{I}$ is the index set of all grid points. Finally, we calculate the minimal harvested DC power among all receivers for each grid point, which results in a set of solutions $\{P_{\text{dc},i}| \forall i \in \mathcal{I}\} $. The result of the exhaustive search is obtained by taking the maximum value among these solutions, i.e., $P^\star_{\text{dc}} = \underset{i}{\text{max}} \{P_{\text{dc},i}\}$. In our simulation, the grid resolution is chosen as $\xi=0.001$. The corresponding relative difference between the results at the optimal point of exhaustive search and its adjacent points is with the magnitude of $10^{-7}$, which shows a high accuracy of the chosen resolution.

In Fig. \ref{fig:converge} the convergence behavior of the proposed iterative algorithm with 5 and 10 receivers is depicted. For each scenario the iterative algorithm is tested with two initial points. Specifically, one initial point is chosen as the point that is very close to a WPT receiver, denoted as the ``ini-bad'' in the figure, and another initial point is chosen as the point that is close to the geometry center point of all WPT receivers, denoted as the ``ini-good'' in the figure. It is observed that the results of both initializations of the proposed iterative algorithm converge to the same stable point which is very close to the global optimum\footnote{Note that because they convergence to the same point, in the enlarged figures only one iterative curve and the exhaustive curve are shown due to the overlapping of the iterative curves.}. Moreover, the results also show that with a better initial point, i.e., a point that is close to the optimal point, the iterative algorithm converges within fewer iterations. Nevertheless, the algorithm convergences within less 7 iterations even with a very bad initial point, e.g., a point that is very close to one of the receivers. This shows a good applicability of the proposed iterative algorithm. }

\section{Conclusion}
In this work, we addressed the convex property of a nonlinear WPT  EH model. We showed that the harvested DC power via the nonlinear model is convex in the reciprocal of the power of the received RF signal. This result indicates that
the harvested DC power is convex in the path-loss of the RF signal transmitting link, which facilitates WPT network designs, i.e., resource allocation, WPT devices positioning.
As an example, we provide a case study of applying the proved convexity in a WPT transmitter positioning problem.
Owning to the convexity, we approximate the non-convex problem and solve it in a interactive manner. The simulation results confirm the converging speed of the interactive algorithm as well as its performance in comparison to the  exhaustive search.

\bibliographystyle{IEEEtran}

\end{document}